\documentclass[smallextended]{svjour3}       
\smartqed  
\usepackage{amssymb}
\usepackage{color}
\usepackage[colorinlistoftodos,textsize=small]{todonotes}

\usepackage{framed}

\newtheorem{algorithm}{Algorithm}

\usepackage{enumitem}
\setlist[description]{font=\itshape}

\newcommand{\change}[1]{\textcolor{black}{#1}}
\definecolor{darkred}{rgb}{0.7,0.1,0.1}

\usepackage{tikz}
\usetikzlibrary{calc}
\usetikzlibrary{snakes}
\usetikzlibrary{trees}
\tikzstyle{vertex}=[circle,fill=black!0,minimum size=4pt,inner sep=0pt]
\tikzstyle{smallvertex}=[circle,fill=black!0,minimum size=3pt,inner sep=0pt]

\begin{document}

\title{Do branch lengths help to locate a tree in a phylogenetic network?}


\author{Philippe Gambette \and
        Leo van Iersel \and
        Steven Kelk \and
        Fabio Pardi \and
        Celine Scornavacca 
}


\institute{
            P. Gambette \at
              Universit\'e Paris-Est, LIGM (UMR 8049), CNRS, ENPC, ESIEE Paris, UPEM, F-77454, Marne-la-Vall\'ee, France\\
              \email{philippe.gambette@u-pem.fr}
              \and  
            L. van Iersel \at        
              Delft Institute of Applied Mathematics, Delft University of Technology\\
              Postbus 5031,2600 GA Delft, The Netherlands\\
              \email{l.j.j.v.iersel@gmail.com}
              \and         
            S. Kelk \at
              Department of Data Science and Knowledge Engineering (DKE)\\
              Maastricht University,  P.O. Box 616, 6200 MD, Maastricht, The Netherlands\\
              \email{steven.kelk@maastrichtuniversity.nl}
              \and
                \emph{Corresponding author:} F. Pardi \at
              Institut de Biologie Computationnelle (IBC)\\
              Laboratoire d'Informatique, de Robotique et de Micro\'electronique de Montpellier (LIRMM)\\ 
              CNRS, Universit\'e de Montpellier, France\\
              \email{pardi@lirmm.fr}
              \and
            C. Scornavacca \at
              Institut de Biologie Computationnelle (IBC)\\
              Institut des Sciences de l'Evolution, CC 064 \\
              Place Eug\`ene Bataillon, Montpellier, France\\
              \email{celine.scornavacca@umontpellier.fr}   
}

\date{Received: date / Accepted: date}

\maketitle

\newpage

\begin{abstract} 

Phylogenetic networks are increasingly used in evolutionary biology
to represent the history of species that have undergone reticulate 
events such as horizontal gene transfer, hybrid speciation and recombination.
One of the most fundamental questions that arise in this context is whether 
the evolution of a 
gene with one copy in all species
can be explained by a given network.
In mathematical terms, this is often 
translated in the following way:
is a given phylogenetic tree contained in a given phylogenetic network?
Recently this \emph{tree containment} problem has been widely investigated from 
a computational perspective, but most studies have only focused on the 
topology of the phylogenies, ignoring a piece of information that, 
in the case of phylogenetic trees, is routinely inferred by evolutionary 
analyses: branch lengths. These measure the amount of change (e.g., nucleotide substitutions) 
that has occurred along each branch of the phylogeny. 
Here, we study a number of versions of the tree containment problem that explicitly
account for branch lengths. 
We show that, although length information has the potential to locate more precisely 
a tree within a network, the problem is computationally hard in its most general form.
On a positive note, for a number of special cases of biological relevance,
we provide algorithms that solve this problem efficiently. 
This includes the case of networks of limited complexity, 
for which it is possible to recover, among the trees contained by the network  
with the same topology as 
the input tree, the closest one in terms of branch lengths.


\end{abstract}

\section{Introduction}

The last few years have witnessed a growing appreciation of reticulate 
evolution -- that is, cases where the history of a set of taxa (e.g., species, 
populations or genomes) cannot be accurately represented as a phylogenetic 
tree \cite{Doolittle1999,bapteste2013networks}, because of events causing
inheritance from more than one ancestor.
Classic examples of such reticulate events 
are hybrid speciation \cite{mallet2007hybrid,nolte2010understanding,abbott2013hybridization},
horizontal gene transfer \cite{Boto2010,hotopp2011horizontal,zhaxybayeva2011lateral} 
and recombination \cite{posada2002recombination,vuilleumier2015contribution}.
Inferring the occurrence of these events in the past is a crucial step towards 
tackling major biological issues, for example to understand recombinant aspects 
of viruses such as HIV \cite{rambaut2004causes}, or characterizing the mosaic 
structure of plant genomes.


Reticulate evolution is naturally
represented by \emph{phylogenetic networks} -- mathematically, simple generalizations
of phylogenetic trees, where some nodes are allowed to have multiple direct 
ancestors \cite{book,morrison_book}.
Currently, much of the mathematical and computational literature on this subject
focuses solely on the topology of phylogenetic networks \cite{HusonScornavacca11},
namely not taking into account branch length information.
This information -- a measure of elapsed time, or 
of change that a species or gene has undergone along a branch --
is usually estimated when inferring phylogenetic trees, and it
may have a big impact on the study of reticulate evolution as well.


For example, in the literature investigating hybridization in the presence of 
incomplete lineage sorting, the branch lengths of a phylogenetic network are 
the key parameters to calculate the probability of observing a gene tree, 
and thus to determine the likelihood of the network 
\cite{meng2009detecting,yu2012probability}.
Moreover, accurate estimates of branch lengths in the gene trees are known to 
improve the accuracy of the inferred network \cite{Kubatko2009,YDLN2014}.
Similarly, for another large class of methods for network reconstruction,
otherwise indistinguishable network scenarios can become distinguishable,
if branch lengths are taken into account \cite{pardi2015reconstructible}.
\change{The precise meaning of branch lengths is often context-dependent,
ranging from expected number of substitutions per site, generally adopted
in molecular phylogenetics, to a measure of the probability of coalescence,
often adopted for smaller timescales where incomplete lineage sorting is common, 
to the amount of time elapsed. In the last case, we may expect the 
phylogeny (network or tree) to be ultrametric, that is to have all its leaves 
at the same distance from the root \cite{chan2006reconstructing,bordewich2016algorithm}.}


In this paper, we explore the impact of branch lengths on a fundamental 
question about phylogenetic networks: the \emph{tree containment} problem.
Informally (formal definitions will be given in the next section), this
problem involves determining whether a given phylogenetic tree is 
contained, or \emph{displayed}, by a given phylogenetic network, 
and in the positive case, locating this tree within the network. 
Biologically, this means understanding whether a gene
-- whose phylogenetic history is well-known -- is consistent 
with a given phylogenetic network, and understanding from which
ancestor the gene was inherited at each reticulate event. 
From a computational perspective, the tree containment problem lies
at the foundation of the reconstruction of phylogenetic networks. 
In its classic version, where only topologies are considered, 
the problem is NP-hard \cite{kanj2008seeing}, but for some specific
classes of networks it can be solved in polynomial time \cite{ISS2010b}.

Intuitively, an advantage of considering branch lengths 
is that it should 
allow one to locate more precisely a gene history within a network, 
and, more generally, it should give more specific answers 
to the tree containment problem. For example, whereas a tree topology
may be contained in multiple different locations inside a network
\cite{cordue2014phylogenetic}, this will happen much more rarely when
branch lengths are taken into account (see, e.g., $T_1$ in 
Fig.\ \ref{fig:example_display}). 
Similarly, some genes may only be detected to be inconsistent with a 
network when the branch lengths of their phylogenetic trees 
are considered (see, e.g.,\ $T_2$ in Fig.\ \ref{fig:example_display}).
%
In practice, some uncertainty in the branch length estimates is to be
expected, which implies that deciding whether a tree is contained
in a network will depend on the confidence in these estimates
(e.g., $T_2$ in Fig.\ \ref{fig:example_display} is only displayed by $N$
if we allow its branch lengths to deviate by 2 or more units from their 
specified values).



While the possibility of having more meaningful answers to a computational 
problem is certainly an important advantage, another factor to consider
is the complexity of calculating its solutions. 
It is known that adding constraints on branch lengths can lead to polynomial 
tractability of other problems in phylogenetics that would otherwise be NP-complete \cite{DSGSRB2010}.
In this paper, we will show a number of results on the effect of taking into account branch 
lengths on the computational complexity of the tree containment problem.
We first introduce the necessary mathematical preliminaries 
(Sec.\ \ref{sec:prel}), including a formal definition of the main problem that 
we consider (\textsc{Tree Containment with Branch Lengths -- TCBL}), and of some
variations of this problem accounting for the fact that branch lengths are 
usually only imprecise estimates of their true values (\textsc{relaxed-TCBL} 
and \textsc{closest-TCBL}).
We then show a number of hardness (negative) results 
for the most general versions of these problems (Sec.\ \ref{sec:neg}), followed 
by a number of positive results (Sec.\ \ref{sec:pos}). 
Specifically, a suite of polynomial-time, pseudo-polynomial time and fixed 
parameter tractable algorithms that solve the problems above
for networks of limited complexity (measured by their level \cite{choy2005computing,jansson2006inferring}; definition below)
and containing no unnecessary complexity (no redundant blobs \cite{trinets1};
also defined below).

\begin{figure}
 \begin{center}

\begin{tabular}{ccccc}
$N$ & ~ & $T_1$ & ~ & $T_2$
\\

& ~ & ~ & ~ &
\\

\begin{tikzpicture}
\node[vertex,draw] (rootp) at  (0,3) {};
\node[vertex,draw] (root) at (0,2.5) {};
\node[vertex,draw] (abc) at  (-.5,2) {};
\node[vertex,draw] (bcd) at  (.5, 2) {};
\node[vertex,draw] (hbc) at  (0,1.5) {};

\node[vertex,draw] (bc) at ( 0,  1) {};
\node[vertex,draw] (ab) at (-1.5,1) {};
\node[vertex,draw] (cd) at ( 1.5,1) {}; 

\node[vertex,draw] (hb) at (-.75,.5) {};
\node[vertex,draw] (hc) at ( .75,.5) {};

\node[vertex,draw] (a) at (-2,0) [label=below:$a$] {};
\node[vertex,draw] (b) at (-.75,0) [label=below:$b$] {};
\node[vertex,draw] (c) at ( .75,0) [label=below:$c$] {};
\node[vertex,draw] (d) at ( 2,0) [label=below:$d$] {};

\draw [very thick] (rootp) -- (root);
\draw [very thick] (root) -- (abc);
\draw [very thick] (root) -- (bcd);
\draw              (abc) -- (hbc);
\draw [very thick] (bcd) -- (hbc);
\draw [very thick] (abc) -- (ab) node [midway, fill=none, xshift=-.2cm, yshift=.2cm] {$2$};
\draw [very thick] (bcd) -- (cd) node [midway, fill=none, xshift=.2cm, yshift=.2cm] {$2$};
\draw [very thick] (hbc) -- (bc);
\draw              (bc) -- (hb);
\draw [very thick] (bc) -- (hc);
\draw [very thick] (ab) -- (a) node [midway, fill=none, xshift=-.25cm, yshift=.1cm] {$2$};
\draw [very thick] (ab) -- (hb);
\draw              (cd) -- (hc);
\draw [very thick] (hb) -- (b);
\draw [very thick] (hc) -- (c);
\draw [very thick] (cd) -- (d) node [midway, fill=none, xshift=.25cm, yshift=.1cm] {$2$};

\end{tikzpicture}
& ~ &
\begin{tikzpicture}
\node[vertex,draw] (rootp) at  (0,2.1) {}; 
\node[vertex,draw] (root) at (0,1.7) {};
\node[vertex,draw] (bcd) at  (.7, 1.33) {};
\node[vertex,draw] (ab) at (-.915,.67) {};

\node[vertex,draw] (a) at (-1.33,0) [label=below:$a$] {};
\node[vertex,draw] (b) at (-.5,0) [label=below:$b$] {};
\node[vertex,draw] (c) at ( .5,0) [label=below:$c$] {};
\node[vertex,draw] (d) at ( 1.33,0) [label=below:$d$] {};

\draw (rootp) -- (root);
\draw (root) -- (ab) node [midway, fill=none, xshift=-.2cm, yshift=.2cm] {$3$};
\draw (root) -- (bcd);
\draw (ab) -- (a) node [midway, fill=none, xshift=-.2cm, yshift=.1cm] {$2$};
\draw (ab) -- (b) node [midway, fill=none, xshift=.2cm, yshift=.1cm] {$2$};
\draw (bcd) -- (c) node [midway, fill=none, xshift=-.2cm, yshift=0cm] {$4$};
\draw (bcd) -- (d) node [midway, fill=none, xshift=.2cm, yshift=.1cm] {$4$};
\end{tikzpicture}
& ~ &
\begin{tikzpicture}
\node[vertex,draw] (rootp) at  (0,2.1) {}; 
\node[vertex,draw] (root) at (0,1.7) {};
\node[vertex,draw] (bc) at  (-.5, .4) {};
\node[vertex,draw] (abc) at (-.915,.67) {};

\node[vertex,draw] (a) at (-1.33,0) [label=below:$a$] {};
\node[vertex,draw] (b) at (-.66,0) [label=below:$b$] {};
\node[vertex,draw] (c) at (-.2,0) [label=below:$c$] {};
\node[vertex,draw] (d) at (.66,0) [label=below:$d$] {};

\draw (rootp) -- (root);
\draw (root) -- (abc) node [midway, fill=none, xshift=-.2cm, yshift=.2cm] {$3$};
\draw (root) -- (d) node [midway, fill=none, xshift=.2cm, yshift=.1cm] {$5$};
\draw (abc) -- (bc);
\draw (bc) -- (b);
\draw (bc) -- (c);
\draw (abc) -- (a) node [midway, fill=none, xshift=-.2cm, yshift=0.1cm] {$2$};
\end{tikzpicture}

\end{tabular}

\end{center}
\caption{\textbf{Toy example on the impact of branch lengths on locating a tree within a network.}
If lengths are not taken into account, both $T_1$ and $T_2$ are displayed by
$N$. Moreover, locating uniquely $T_1$ within $N$ is not possible:
there are 4 switchings \change{(formally defined in the Preliminaries)} of $N$ for $T_1$,
and 3 different ways to locate (i.e.\emph{ images} of) $T_1$ within $N$.
If instead lengths are taken into account, the only image of $T_1$ within $N$ is the one highlighted in 
bold, and $T_2$ is not displayed by $N$ (in fact, the tree displayed by $N$ \change{isomorphic to} 
$T_2$ has significantly different branch lengths).
\emph{Note}: branches with no label are assumed to have length 1.}
\label{fig:example_display}
\end{figure}
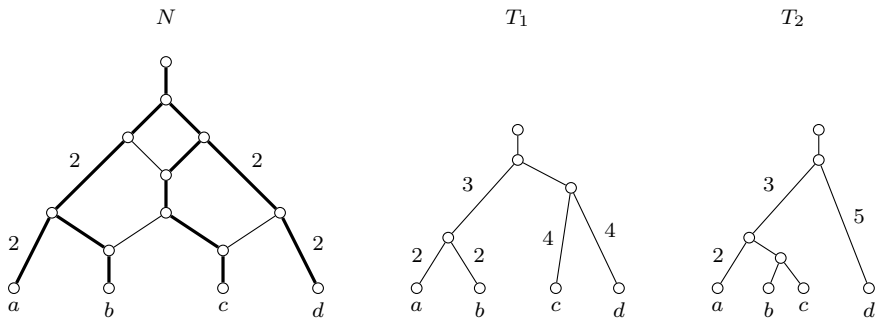


\section{Preliminaries} \label{sec:prel}

We define a \emph{phylogenetic network} on $X$ as a rooted directed
acyclic graph with exactly one vertex of indegree 0 (the \emph{root}), 
with no vertices with indegree and outdegree 1, and whose outdegree 0 vertices 
(the \emph{leaves}) are bijectively labeled by the elements of $X$ (the \emph{taxa}).
A \emph{phylogenetic tree} is a phylogenetic network whose underlying undirected graph has no cycles.
We consider phylogenetic networks
(and thus trees) where each arc has an associated length. Formally,
given an arc $(u,v)$ of a phylogenetic network $N$,
its \emph{length} $\lambda_N(u,v)$ is a positive integer, \change{i.e. strictly greater than zero}.
In this paper we will use the terms ``arc lengths'' and  ``branch lengths'' interchangeably.

A phylogenetic tree or network is \emph{binary} if
all vertices have indegree 1 and outdegree 2 (\emph{bifurcations}),
indegree 2 and outdegree 1 (\emph{reticulations}),
indegree 0 and outdegree 1 (root)
or indegree 1 and outdegree 0 (leaves).
For example, all networks and trees in Fig.\ \ref{fig:example_display} are binary.

A \emph{biconnected component} is a maximal connected subgraph that 
remains connected after removal of any one vertex.
A \emph{blob} of a phylogenetic network $N$ is a biconnected component in which the undirected graph
underpinning the biconnected component contains at least one cycle. 
Note that if a biconnected component of $N$ is not a blob, then it is simply a cut arc
(i.e., an arc whose removal disconnects $N$). 
The \emph{level} of a binary phylogenetic network $N$ is the maximum number of 
reticulations in any blob of $N$.
An \emph{outgoing arc} of a blob $B$ is an arc $(u,v)$ such that $u$ is in $B$ but $v$ is not. 
An \emph{incoming arc} $(u,v)$ of $B$ is such that $v$ is in $B$ but $u$ is not.
Note that a blob has at most one incoming arc.
A blob is \emph{redundant} if it has fewer than two outgoing arcs
(i.e., one outgoing arc if the network is binary).
As an example of these notions, the network $N$ in Fig.~\ref{fig:example_display}
contains only one blob, which has 4 outgoing arcs and is thus non-redundant.
Because this blob has 3 reticulations, $N$ is level-3.

\change{Two phylogenetic trees $T_1$ and $T_2$ are said to be \emph{isomorphic} or to have 
the same \emph{topology}, if there exists a one-to-one mapping from the nodes of $T_1$ 
onto the nodes of $T_2$ preserving leaf labels and descendancy (but not arc lengths).}

Given a phylogenetic tree $T$ and a phylogenetic network $N$ 
whose leaves are labeled bijectively by the same set $X$,
we say that \emph{$T$ is displayed by $N$
taking into account lengths},
if $T$ can be obtained from $N$ in the following way:
\begin{itemize}

\item[$\bullet$] for each reticulation,
remove all incoming arcs except one;
the tree obtained after this process 
is called a \emph{switching \label{def:switchings}} of $N$; 
\item[$\bullet$] repeat as long as possible the following \emph{dummy leaf deletions}:
for each leaf not labeled by an element of $X$, delete it;
\item[$\bullet$] repeat as long as possible the following \emph{vertex smoothings}:
for each vertex $v$ with exactly one parent $p$ and one child $c$,
replace it with an arc from $p$ to $c$, with $\lambda_N(p,c)=\lambda_N(p,v)+\lambda_N(v,c)$.

\end{itemize}
In the following, we sometimes only say that \emph{$T$ is displayed by $N$}
(with no mention of lengths) 
to mean that arc lengths are disregarded, 
and only topological information is taken into account.

Note that 
$N$ displays $T$ taking into account lengths if and only if 
there exists a subtree $T'$ of $N$
with the same root as $N$ such that 
$T$ can be obtained by repeatedly applying vertex smoothings to $T'$. 
In this case $T'$ is said to be the \emph{image} of $T$. 
There is a natural injection from the vertices of $T$
to the vertices of $T'$, so the definition of image extends naturally to 
any subgraph of $T$. In particular, the image of any arc in $T$ is a path in $N$.
Note that $T$ can potentially have many images in $N$, but 
for a switching $S$ of $N$, the image of $T$ within $S$, if it exists, is unique.
As an example of these notions, consider again Fig.~\ref{fig:example_display},
where $N$ displays both $T_1$ and $T_2$, but only $T_1$ if lengths are taken into account.
The part of $N$ in bold is both a switching and an image of $T_1$
(as no dummy leaf deletions are necessary in this case).

Finally, is worth noting that, in this paper, if $N$ displays $T$ taking into account lengths,
then the image of the root of $T$ will always coincide with the root of $N$
(no removal of vertices with indegree 0 and outdegree 1 is applied to obtain $T$).
The biological justification for this is that 
trees and networks are normally rooted using an outgroup, which is sometimes omitted from the phylogeny; 
if arc lengths are taken into account, then the length of the path to the root of $N$ 
in a tree displayed by $N$ conveys the information regarding the distance from the outgroup. 
(See also \cite{pardi2015reconstructible} for a full discussion about this point.)

In this paper, we consider the following problem:

\begin{framed} 
\begin{problem}
\textsc{Tree Containment with Branch Lengths} (TCBL)
\begin{description}
\item[Input:] A phylogenetic network $N$
and a phylogenetic tree $T$ on the same set $X$, and
both with \change{positive} integer arc lengths.
\item[Output:] YES if $T$ is displayed by $N$
taking into account lengths, NO otherwise.
\end{description}
\end{problem}
\end{framed} 

We also consider two variations of TCBL 
seeking trees displayed by $N$ that are allowed to somehow deviate from the query tree,
to account for uncertainty in the branch lengths of the input tree.
The first of these two problems aims to determine the existence of a tree 
displayed by $N$, whose branch lengths fall within a specified (confidence) interval. 

\begin{framed} 
\begin{problem}
\textsc{relaxed-TCBL}
\begin{description}
\item[Input:] A phylogenetic network $N$ with  \change{positive}  integer arc lengths,
and a phylogenetic tree $T$, whose arcs are labelled by two  \change{positive}  integers $m_T(a)$ and $M_T(a)$, representing respectively the minimum and the maximum arc length. Both $N$ and $T$ are on the same set $X$.
\item[Output:] YES if and only if there exists a tree $\widetilde{T}$ displayed by  $N$, \change{isomorphic to} 
$T$, and such that, for each arc $a$ of $T$:
\[
  \lambda_{\widetilde{T}}(\tilde{a}) \in \left[ m_T(a) , M_T(a) \right],
\]
where $\tilde{a}$ denotes the arc in $\widetilde{T}$ that corresponds to $a$ in $T$.
\end{description}
\end{problem}
\end{framed}

The second variation of TCBL we consider here, seeks
-- among all trees displayed by the network, \change{and that are isomorphic to} 
the input tree $T$ 
-- one that is closest to $T$, in terms of the maximum difference between branch 
lengths.
\change{There are several other alternative choices for defining the ``closest'' tree to $T$,
for example if distance is measured in terms of the average difference between branch lengths.
Later on, we will see that our results on this problem also apply to many of these alternative formulations
(see Theorem \ref{thm:reformulations}).}

\begin{framed} 
\begin{problem} \textsc{closest-TCBL}
\begin{description}
\item[Input:] A phylogenetic network $N$ and a phylogenetic tree $T$ on the same set $X$, and
both with  \change{positive}  integer arc lengths.
\item[Output:] A tree $\widetilde{T}$ displayed by $N$, \change{isomorphic to} 
$T$, that minimizes
\[
  \max \left| \lambda_T(a) - \lambda_{\widetilde{T}}(\tilde{a}) \right|,
\]
where the $\max$ is over any choice of an arc $a$ in $T$, and  
$\tilde{a}$ denotes the arc in $\widetilde{T}$ that corresponds to $a$ in $T$.
If no tree  \change{isomorphic to} 
$T$ is displayed by  $N$, then report FAIL.
\end{description}
\end{problem}
\end{framed}

\change{Note that all problems in this paper involving
positive integer arc lengths are equivalent to
problems where arc lengths are 
positive rational numbers: it 
suffices to multiply those rational numbers
by the least common denominator of the fractions
corresponding to these numbers in order to get
integers.}

We conclude with some definitions concerning computational complexity. An NP-complete decision problem
that includes numbers in the input may or may not permit a \emph{pseudo-polynomial time} algorithm. This is
an algorithm which runs in polynomial time if the numbers in the input are encoded in unary, rather than binary. 
Formally speaking such algorithms are not polynomial time, since unary encodings artificially inflate
the size of the input. Nevertheless, a pseudo-polynomial time algorithm has the potential to run quickly if the
numbers in the input are not too large. An NP-complete problem with numbers in the input is said to be \emph{strongly} 
NP-complete if it remains NP-complete even under unary encodings of the numbers. Informally, such problems remain
intractable even if the numbers in the input are small. An NP-complete problem is \emph{weakly} NP-complete if
it is NP-complete when the numbers are encoded in binary. Summarizing, if one shows that a weakly NP-complete problem also
permits a pseudo-polynomial time algorithm, then (under standard complexity assumptions) this excludes strong NP-completeness.
Similarly, demonstrating strong NP-completeness excludes (under standard complexity assumptions) the existence of a 
pseudo-polynomial time algorithm. We refer to Garey and Johnson \cite{GareyJohnson1979} for 
formal definitions.

On a slightly different note, an algorithm is said to be \emph{fixed parameter tractable} (FPT) if it runs 
in time $O( f(k) \cdot poly(n) )$ where $n$ is the size of the input, $k$ is some parameter of the input (in this article: the
level of the network) and $f$ is some computable function that depends only on $k$. An FPT algorithm for an NP-complete
problem has the potential to run quickly even when $n$ is large, as long as the parameter $k$ is small, 
\change{for example when $f$ is a $c^k$ function, where
$c$ is a small constant greater than 1.}
We refer to \cite{downey2013fundamentals,Gramm2008} for more background on FPT algorithms.

\section{Negative results\label{sec:neg}}

\subsection{Strong NP-completeness\label{sec:strongnp}}

\begin{theorem}\label{TCBLstrong}
TCBL is strongly NP-complete, even when the phylogenetic tree $T$ and
the phylogenetic network $N$ are binary.
\end{theorem}
\begin{proof}
We reduce to TCBL the following \textsc{3-Partition} problem,
which is strongly NP-complete~\cite{GareyJohnson1975}:
\begin{description}
\item[Input:] an integer $\Sigma$ and
a multiset $S$ of $3m$ positive integers $n_i$ in $]\Sigma/4,\Sigma/2[$
such that $m\Sigma = \sum\limits_{i\in [1..3m]} n_i$.
\item[Output:] YES if $S$ can be partitioned
into $m$ subsets of elements $S_1, S_2, \ldots, S_m$
each of size 3,
such that the sums of the numbers in each subset are all equal;
NO otherwise.
\end{description}

Let us consider a multiset $S$ containing $3m$ positive integers $n_i$
which have sum $m\Sigma$.

We build a phylogenetic tree $T$ in the following way.
We first build a directed path 
containing $m+2$ vertices,
whose arcs all have length 1.
We call its initial vertex $\rho$, its final vertex $b_0$, and
the ancestors of $b_0$, from the parent of $b_0$ to the child of $\rho$
are called $v_1$ to $v_m$.
Then, to each of the $m$ vertices $v_i$ for $i\in [1..m]$
on this directed path, from bottom to top,
we add an arc of length $L=\Sigma+6m^2-3m+1$ to a child,
called $b_i$.

We now build a phylogenetic network $N$ in the following way.
We start by creating a copy of $T$ but for
each $i\in[1..m]$ we remove the arc $(v_i,b_i)$
and replace it by an arc of length 1 from $v_i$ to a new vertex
$r_1^i$ (see Figure~\ref{fig:TCBLglobal}).
Then we create $3m$ subnetworks called $B_k$, for $k \in [1..3m]$,
as described in Figure~\ref{fig:TCBLblock}.
For ease of notation, we consider that vertex $p^2_k$
is also labeled $p^1_k$ and $c^2_k$ is also labeled
$c^1_k$ for any $k \in [1..3m]$.
Finally, we add arcs $(b_k^i,r_{k+1}^i)$ of length 1
for each $k \in [1..3m-1]$ and $i \in [1..m]$
(to connect each $B_k$ with $B_{k+1}$) and
arcs of length 1 from $b_{3m}^i$ to $b_i$ for each $i \in [1..m]$
to obtain $N$.
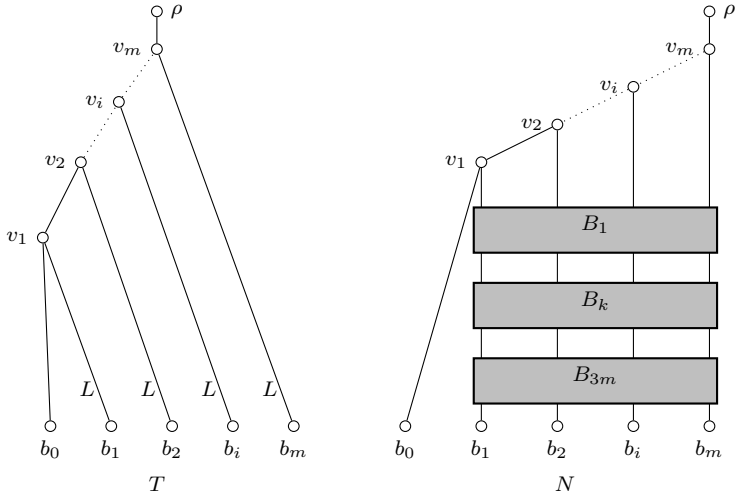
\begin{figure}[!ht]
\centering
\begin{tabular}{ccc}
\begin{tikzpicture}
\node[vertex,draw] (rootparent) at (3,5.5) [label=right:$\rho$] {};
\node[vertex,draw] (root) at (3,5) [label=left:$v_m$] {};
\node[vertex,draw] (v3) at (2.5,4.3) [label=left:$v_i$] {};
\node[vertex,draw] (v2) at (2,3.5) [label=left:$v_2$] {};
\node[vertex,draw] (v1) at (1.5,2.5) [label=left:$v_1$] {};
\node[vertex,draw] (b0) at (1.6,0) [label=below:$b_{0}$] {};
\node[vertex,draw] (b1) at (2.4,0) [label=below:$b_1$] {};
\node[vertex,draw] (b2) at (3.2,0) [label=below:$b_2$] {};
\node[vertex,draw] (b3) at (4,0) [label=below:$b_i$] {};
\node[vertex,draw] (b4) at (4.8,0) [label=below:$b_m$] {};

\node (vb1) at (2.4,0.5) [label=left:$L$] {};
\node (vb2) at (3.2,0.5) [label=left:$L$] {};
\node (vbi) at (4,0.5) [label=left:$L$] {};
\node (vbm) at (4.8,0.5) [label=left:$L$] {};

\draw (rootparent) -- (root);
\draw (v1) -- (v2);
\draw [dotted] (v2) -- (v3);
\draw [dotted] (v3) -- (root);
\draw (b0) -- (v1);
\draw (b1) -- (v1);
\draw (b2) -- (v2);
\draw (b3) -- (v3);
\draw (b4) -- (root);
\end{tikzpicture}
& ~ &
\begin{tikzpicture}
\node[vertex,draw] (rootparent) at (6,5.5) [label=right:$\rho$] {};
\node[vertex,draw] (root) at (6,5) [label=left:$v_{m}$] {};
\node[vertex,draw] (v3) at (5,4.5) [label=left:$v_{i}$] {};
\node[vertex,draw] (v2) at (4,4) [label=left:$v_{2}$] {};
\node[vertex,draw] (v1) at (3,3.5) [label=left:$v_{1}$] {};
\node[vertex,draw] (b0) at (2,0) [label=below:$b_0$] {};


\node[vertex,draw] (c1) at (3,0) [label=below:$b_1$] {};
\node[vertex,draw] (c2) at (4,0) [label=below:$b_2$] {};
\node[vertex,draw] (c3) at (5,0) [label=below:$b_i$] {};
\node[vertex,draw] (c4) at (6,0) [label=below:$b_m$] {};

\draw (rootparent) -- (root);
\draw [dotted] (root) -- (v3);
\draw [dotted] (v3) -- (v2);
\draw (v2) -- (v1);
\draw (v1) -- (b0);

\draw (root) -- (c4);
\draw (v3) -- (c3);
\draw (v2) -- (c2);
\draw (v1) -- (c1);

\draw [thick,fill=lightgray] (2.9,2.9) rectangle (6.1,2.3);
\draw [thick,fill=lightgray] (2.9,1.9) rectangle (6.1,1.3);
\draw [thick,fill=lightgray] (2.9,0.9) rectangle (6.1,0.3);
\node at (4.5,3) [label=below:$B_{1}$] {};
\node at (4.5,2) [label=below:$B_{k}$] {};
\node at (4.5,1) [label=below:$B_{3m}$] {};

\end{tikzpicture}
\\
$T$ & ~ & $N$
\end{tabular}

\caption{\textbf{The tree $T$ and the 
network $N$ used in the proof of Theorem~\ref{TCBLstrong}.}
All arcs are directed downwards. The dotted arcs
represent parts of the network which are not shown
in details but which ensure connectivity.
All arcs incident to leaves $b_i$ of $T$,
for $i \in [1..m]$ have length $L=\Sigma+6m^2-3m+1$;
and remaining arcs of $T$ have length 1. 
All arcs of $N$ have length 1, except in the  $3m$ boxes  $B_k$ 
(see Figure~\ref{fig:TCBLblock}(a) for more details on the content
of those $3m$ boxes $B_k$).}
\label{fig:TCBLglobal}
\end{figure}

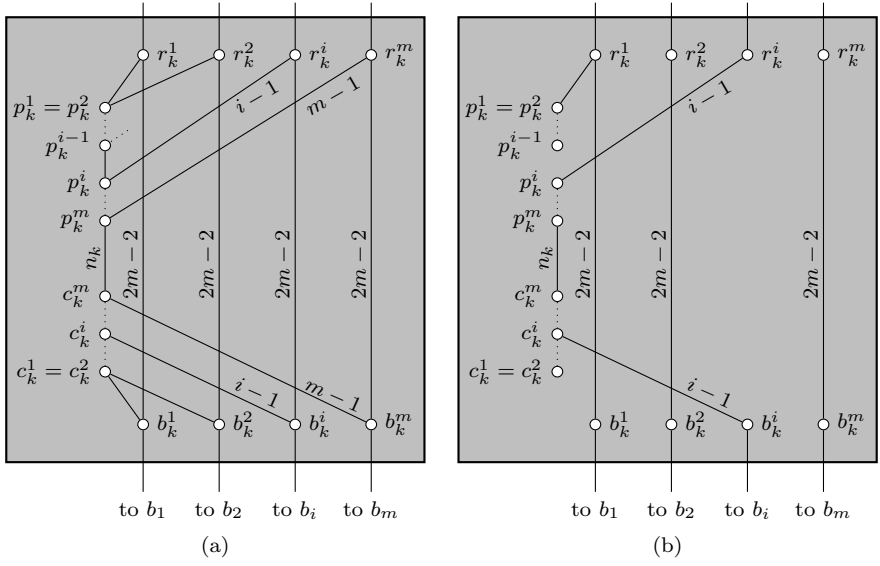
\begin{figure}[!ht]
\centering

\begin{tabular}{cc}
\begin{tikzpicture}
\draw [thick,fill=lightgray] (0.2,6.2) rectangle (5.7,0.3);

\node (root1) at (2,6.5) {};
\node (root2) at (3,6.5) {};
\node (root3) at (4,6.5) {};
\node (root4) at (5,6.5) {};
\node[vertex,draw] (a2) at (1.5,5) [label=left:\footnotesize{$p_k^1=p_k^2$}] {};
\node[vertex,draw] (pim1) at (1.5,4.5) [label=left:\footnotesize{$p_k^{i-1}$}] {};
\node[vertex,draw] (a3) at (1.5,4) [label=left:\footnotesize{$p_k^i$}] {};
\node[vertex,draw] (a4) at (1.5,3.5) [label=left:\footnotesize{$p_k^m$}] {};

\node[vertex,draw] (r1) at (2,5.7) [label=right:\footnotesize{$r_k^1$}] {};
\node[vertex,draw] (r2) at (3,5.7) [label=right:\footnotesize{$r_k^2$}] {};
\node[vertex,draw] (r3) at (4,5.7) [label=right:\footnotesize{$r_k^i$}] {};
\node[vertex,draw] (r4) at (5,5.7) [label=right:\footnotesize{$r_k^m$}] {};

\node (c1) at (2,0) [label=below:to $b_1$] {};
\node (c2) at (3,0) [label=below:to $b_2$] {};
\node (c3) at (4,0) [label=below:to $b_i$] {};
\node (c4) at (5,0) [label=below:to $b_m$] {};
\node (cc1) at (2,-0.2) {};
\node (cc2) at (3,-0.2) {};
\node (cc3) at (4,-0.2) {};
\node (cc4) at (5,-0.2) {};

\node[vertex,draw] (l1) at (2,0.8) [label=right:\footnotesize{$b_k^1$}] {};
\node[vertex,draw] (l2) at (3,0.8) [label=right:\footnotesize{$b_k^2$}] {};
\node[vertex,draw] (l3) at (4,0.8) [label=right:\footnotesize{$b_k^i$}] {};
\node[vertex,draw] (l4) at (5,0.8) [label=right:\footnotesize{$b_k^m$}] {};

\node[vertex,draw] (b2) at (1.5,1.5) [label=left:\footnotesize{$c_k^1=c_k^2$}] {};
\node[vertex,draw] (b3) at (1.5,2) [label=left:\footnotesize{$c_k^i$}] {};
\node[vertex,draw] (b4) at (1.5,2.5) [label=left:\footnotesize{$c_k^m$}] {};

\node(e23) at (1.35,3.4) [label=below:\rotatebox{90}{\footnotesize{$n_k$}}] {};
\node(ext) at (1.85,3.6) [label=below:\rotatebox{90}{\footnotesize{$2m-2$}}] {};
\node(ext) at (2.85,3.6) [label=below:\rotatebox{90}{\footnotesize{$2m-2$}}] {};
\node(ext) at (3.85,3.6) [label=below:\rotatebox{90}{\footnotesize{$2m-2$}}] {};
\node(ext) at (4.85,3.6) [label=below:\rotatebox{90}{\footnotesize{$2m-2$}}] {};
\node[label=below:\rotatebox{30}{\footnotesize{$i-1$}}](ei) at (3.5,5.6) {};
\node[label=below:\rotatebox{30}{\footnotesize{$m-1$}}](ei) at (4.5,5.65) {};
\node[label=below:\rotatebox{-24}{\footnotesize{$i-1$}}](ei) at (3.5,1.65) {};
\node[label=below:\rotatebox{-24}{\footnotesize{$m-1$}}](ei) at (4.5,1.65) {};

\draw (root1) -- (r1);
\draw (root2) -- (r2);
\draw (root3) -- (r3);
\draw (root4) -- (r4);

\draw (r1) -- (a2);
\draw (r2) -- (a2);
\draw (r3) -- (a3);
\draw (r4) -- (a4);
\draw [dotted] (a2) -- (pim1);
\draw [dotted] (1.8,4.7) -- (pim1);
\draw (pim1) -- (a3);
\draw [dotted] (a3) -- (a4);
\draw (b2) -- (l1);
\draw (l1) -- (cc1);
\draw (l2) -- (cc2);
\draw (l3) -- (cc3);
\draw (l4) -- (cc4);
\draw (l2) -- (b2);
\draw (l3) -- (b3);
\draw (l4) -- (b4);
\draw (l1) -- (r1);
\draw (l2) -- (r2);
\draw (l3) -- (r3);
\draw (l4) -- (r4);
\draw (a4) -- (b4);
\draw [dotted] (b4) -- (b3);
\draw [dotted] (b3) -- (b2);

\end{tikzpicture}
& 
\begin{tikzpicture}
\draw [thick,fill=lightgray] (0.2,6.2) rectangle (5.7,0.3);

\node (root1) at (2,6.5) {};
\node (root2) at (3,6.5) {};
\node (root3) at (4,6.5) {};
\node (root4) at (5,6.5) {};
\node[vertex,draw] (a2) at (1.5,5) [label=left:\footnotesize{$p_k^1=p_k^2$}] {};
\node[vertex,draw] (pim1) at (1.5,4.5) [label=left:\footnotesize{$p_k^{i-1}$}] {};
\node[vertex,draw] (a3) at (1.5,4) [label=left:\footnotesize{$p_k^i$}] {};
\node[vertex,draw] (a4) at (1.5,3.5) [label=left:\footnotesize{$p_k^m$}] {};

\node[vertex,draw] (r1) at (2,5.7) [label=right:\footnotesize{$r_k^1$}] {};
\node[vertex,draw] (r2) at (3,5.7) [label=right:\footnotesize{$r_k^2$}] {};
\node[vertex,draw] (r3) at (4,5.7) [label=right:\footnotesize{$r_k^i$}] {};
\node[vertex,draw] (r4) at (5,5.7) [label=right:\footnotesize{$r_k^m$}] {};

\node (c1) at (2,0) [label=below:to $b_1$] {};
\node (c2) at (3,0) [label=below:to $b_2$] {};
\node (c3) at (4,0) [label=below:to $b_i$] {};
\node (c4) at (5,0) [label=below:to $b_m$] {};
\node (cc1) at (2,-0.2) {};
\node (cc2) at (3,-0.2) {};
\node (cc3) at (4,-0.2) {};
\node (cc4) at (5,-0.2) {};

\node[vertex,draw] (l1) at (2,0.8) [label=right:\footnotesize{$b_k^1$}] {};
\node[vertex,draw] (l2) at (3,0.8) [label=right:\footnotesize{$b_k^2$}] {};
\node[vertex,draw] (l3) at (4,0.8) [label=right:\footnotesize{$b_k^i$}] {};
\node[vertex,draw] (l4) at (5,0.8) [label=right:\footnotesize{$b_k^m$}] {};

\node[vertex,draw] (b2) at (1.5,1.5) [label=left:\footnotesize{$c_k^1=c_k^2$}] {};
\node[vertex,draw] (b3) at (1.5,2) [label=left:\footnotesize{$c_k^i$}] {};
\node[vertex,draw] (b4) at (1.5,2.5) [label=left:\footnotesize{$c_k^m$}] {};

\node(e23) at (1.35,3.4) [label=below:\rotatebox{90}{\footnotesize{$n_k$}}] {};
\node(ext) at (1.85,3.6) [label=below:\rotatebox{90}{\footnotesize{$2m-2$}}] {};
\node(ext) at (2.85,3.6) [label=below:\rotatebox{90}{\footnotesize{$2m-2$}}] {};
\node(ext) at (4.85,3.6) [label=below:\rotatebox{90}{\footnotesize{$2m-2$}}] {};
\node[label=below:\rotatebox{30}{\footnotesize{$i-1$}}](ei) at (3.5,5.6) {};
\node[label=below:\rotatebox{-24}{\footnotesize{$i-1$}}](ei) at (3.5,1.65) {};

\draw (r1) -- (a2);

\draw (root1) -- (r1);
\draw (root2) -- (r2);
\draw (root3) -- (r3);
\draw (root4) -- (r4);

\draw (r3) -- (a3);
\draw [dotted] (a2) -- (pim1);
\draw [dotted] (a3) -- (a4);

\draw (l1) -- (cc1);
\draw (l2) -- (cc2);
\draw (l3) -- (cc3);
\draw (l4) -- (cc4);
\draw (l3) -- (b3);
\draw (l1) -- (r1);
\draw (l2) -- (r2);
\draw (l4) -- (r4);
\draw (a4) -- (b4);
\draw [dotted] (b4) -- (b3);
\draw [dotted] (b3) -- (b2);

\end{tikzpicture}
\\
(a) & (b)
\end{tabular}

\caption{\textbf{The content of the box $B_k$ (a)
and a corresponding switching (b) of the network
of Fig.~\ref{fig:TCBLglobal}.}
All arcs are directed downwards. The dotted arcs
represent parts of the network which are not shown
in details but which ensure connectivity.
All arcs have length 1 except
arcs $(r_k^i,b_k^i)$ for $i \in [1..m]$ which have length 
$2m-2$,
arcs $(r^i_k,p^i_k)$ and $(c^i_k,b^i_k)$, for $i>1$,
which have length $i-1$,
and the arc $(p^m_k,c^m_k)$ with length $n_k$.}
\label{fig:TCBLblock}
\end{figure}

Suppose that $S$ can be partitioned
into $m$ subsets of elements $S_1, S_2, \ldots, S_m$
each of size 3, such that the sums of the numbers in each subset
are all equal to $\Sigma$.
We now prove that this implies that $T$ and $N$ constructed
above constitute a positive instance of TCBL.

For each $n_k$, if it belongs to $S_i$
then we remove from $N$ all arcs $(c^{j}_k,b^j_k)$ for $j \in [1..m]-\{i\}$,
as well as all arcs $(r^j_k,p^{j}_k)$ for $j \in [1..m]-\{i\}-\{1 \textrm{ if }i\neq 2\}$,
the arc $(r^i_k,b^i_k)$, and finally the arc $(p^{i-1}_k,p^i_k)$ if $i\notin\{1,2\}$.
This way, we obtain a switching $T'$ of $N$ for $T$, shown in Figure~\ref{fig:TCBLblock}(b).


In $T'$, the only path from $r^i_k$ to $b^i_k$
goes through the arc $(p^m_k,c^m_k)$ of length $n_k$,
so the total length of this path is $2m-2+n_k$.
For all other $S_j$, $j\in [1..k]-\{i\}$, the only directed
path from $r_k^i$ to $b_k^i$ is an arc of length $2m-2$.
Thanks to the arcs $(b^j_k,r^j_{k+1})$, for $j\in [1,m]$
a unique path can be found in $T'$
from $v_j$ to $b_j$. We can check that the lengths
of the arcs of $T$ leading to $b_i$ with $i \in [1..m]$
are consistent with the lengths of these paths:
the latter have all length
$3m((2m-2)+1)+(\sum\limits_{n_k \in S_i} n_k)+1=\Sigma+6m^2-3m+1$.
Furthermore, all other arcs of $T$ (on the path from
$\rho$ to $b_0$) are also present in $T'$ with
the same configuration and length,
meaning that, 
as we wished to prove,
$T$ is displayed by $N$ taking into account lengths.

We now focus on the converse, supposing
that the tree $T$ is displayed by  $N$ taking
into account lengths. We first note
that any switching $T'$ of $N$ for $T$ contains
the vertices $b_0$, $v_i$ for $i \in [1..m]$,
$\rho$ and the arcs between these vertices.
Furthermore, $T'$ also contains a path $P_i(T')$
from $v_i$ to $b_i$, for each $i \in [1..m]$,
of length $L$.

\textbf{Claim 1}: For any switching $T'$ of $N$ for $T$,
for any $i \in [1..m]$ and $k \in[1..3m]$,
$r^i_k \in P_i(T')$ and $b^i_k \in P_i(T')$.

We prove it by induction on $k$.
For $k=1$, for all $i\in [1..m]$, vertex $r^i_1$ has indegree 1
and its unique parent is contained in $P_i(T')$
so it is also contained in $P_i(T')$.
As arc $(p^m_1,c^m_1)$ belongs to all paths between
$p^i_1$ and $c^j_1$ for $i,j \in [1..m]$, at most one 
of the paths $P_i(T')$ contains $(p^m_1,c^m_1)$.
If no such path exists then all paths $P_i(T')$
contain arc $(r^i_1,b^i_1)$, so $b^i_1 \in P_i(T')$.
Otherwise, we denote by $P_{i_0}(T')$ the path containing
$(p^m_1,c^m_1)$. All other paths $P_i(T')$ for $i\in [1..m]-i_0$ contain
arc  $(r^i_1,b^i_1)$, so $b^i_1 \in P_i(T')$.
Because 
none of those paths contain $b^{i_0}_1$,
we must have $b^{i_0}_1 \in P_{i_0}(T')$.
Therefore, for all $i \in [1..m]$, $b^i_1 \in P_{i}(T')$.

Supposing vertices $r^i_{k-1}$ and $b^i_{k-1}$
belong to $P_i(T')$ for all $i \in [1..m]$,
we can reproduce the proof above by replacing ``1'' by 
``$k$'' each time we refer to $b^i_1$, $c^i_1$,
$p^i_1$ and $r^i_1$ for any $i \in [1..m]$,
in order to deduce that $r^i_{k}$ and $b^i_{k}$
belong to $P_i(T')$.

\textbf{Claim 2}: For any switching $T'$ of $N$ for $T$,
for any $k \in [1..3m]$,
one of the paths $P_i(T')$ contains arc $(p_k^m,c_k^m)$.

First, using Claim 1, we can consider each
portion of the path $P_i(T')$ from $r^i_k$ to $b^i_k$ in $T'$,
and note that this portion has length $2m-2+n_k$ 
if $P_i(T')$ contains arc $(p_k^m,c^m_k)$, or length $2m-2$ otherwise.

Therefore, supposing by contradiction that
there exists at least one $k_0 \in  [1..3m]$ such that
none of the paths $P_i(T')$ contain arc $(p_{k_0}^m,c_{k_0}^m)$,
then the cumulative length $L_{k_0}$ of the portions of all paths $P_i(T')$
between $r^i_{k_0}$ and $b^i_{k_0}$, for $i\in [1..m]$,
is $m(2m-2)$.
Therefore, summing the lengths of all these portions and
the ones of arcs $(b^i_k,r^i_{k+1})$ between them as well
as the ones of the arcs  $(v_i,r^i_1)$ and
$(b^i_{3m},b_i)$ for any $i\in [1..m]$,
the sum $L'$ of the lengths of all paths $P_i(T')$
for $i \in [1..m]$ is at most
$m + 3m(L_{k_0}+m) + (\sum_{k \in [1..3m]} n_k)-n_{k_0} =
m(6m^2-3m+1+\Sigma)-n_{k_0} = mL-n_{k_0}$.
However, the sum $L_T$ of the lengths of all arcs $(v_i,b_i)$ of $T$
is equal to $mL$ 
so $L' < L_T$, meaning that 
$T$ is not displayed by $N$
taking into account lengths: contradiction.

\textbf{Claim 3}: for any switching $T'$ of $N$ for $T$,
for any $i \in [1..m]$, there are exactly 3
arcs of the form $(p^m_{k},c^m_{k})$ contained in $P_i(T')$.

We suppose by contradiction that there exists $i \in [1..m]$,
and $k_1, k_2, k_3$ and $k_4 \in [1..3m]$
such that $(p^m_{k_1},c^m_{k_1})$,
$(p^m_{k_2},c^m_{k_2})$, $(p^m_{k_3},c^m_{k_3})$
and $(p^m_{k_4},c^m_{k_4})$ are contained in $P_i(T')$.
Then, this path has length at least
$n_{k_1}+n_{k_2}+n_{k_3}+n_{k_4}+3m(2m-2)+3m+1>\Sigma+3m(2m-1)+1$
because $n_i>\Sigma/4$ for all $i\in [1..3m]$.
So $T'$ contains a path from $v_i$ to $b_i$ which
is strictly longer than the arc from $v_i$ to $b_i$ in $T$,
so $T$ is not displayed by  $T'$, nor in $N$: contradiction.

Now, we suppose by contradiction that there exists $i \in [1..m]$
such that $P_i$ contains at most 2 arcs of the form $(p^m_{k},c^m_{k})$.
Then, according to Claim 2,
each of the the remaining $3m-2$ arcs of the form $(p^m_{k},c^m_{k})$
must be contained by one of the remaining $m-1$ paths $P_j$ for
$j \in [1..m]-\{i\}$. So at least one of those paths must contain
strictly more than 3 such arcs,
which contradicts the previous paragraph: contradiction.

Finally, for any switching $T'$ of $N$ for $T$, 
the fact that $T$ is displayed by  $N$ taking into account lengths,
implies that the length of each arc $(v_i,b_i)$ of $T$, $\Sigma+6m^2-3m+1$,
equals the length of each path $P_i(T')$.
Claim 2 and 3 imply that the arcs of the form $(p_k^m,c_k^m)$
are partitioned into the paths $P_i(T')$, with each $P_i(T')$ containing
exactly 3 such arcs. Denoting by $n_{k_i}$, $n_{k'_i}$ and $n_{k''_i}$
the length of such arcs, we obtain that the length of $P_i(T')$
equals $n_{k_i} + n_{k'_i} + n_{k''_i}+6m^2-3m+1$, therefore
$n_{k_i} + n_{k'_i} + n_{k''_i} = \Sigma$, which implies 
that $S$ can be partitioned into $m$ subsets of elements
$S_i=\{n_{k_i}, n_{k'_i}, n_{k''_i}\}$,
such that the sums of the numbers in each subset $S_i$
are all equal to $\Sigma$.

Finally, it is easy to see that the problem is in NP:
a switching $T'$ of the input network $N$ is
a polynomial size certificate of the fact
that the input tree $T$ is contained in $N$. We can
check in polynomial time that $T$ can be obtained
from $T'$ by applying dummy leaf deletions
and vertex smoothings until possible, 
and checking that the obtained tree is isomorphic with $T$.
\qed
\end{proof}

\change{We note that Theorem \ref{TCBLstrong} can be extended to binary tree-sibling 
\cite{cardona2008distance} time-consistent \cite{baroni2006hybrids} networks, 
by multiplying by 2 all arc lengths of the network constructed in the proof 
(in order to keep integer arc lengths even if those arcs are subdivided, which happens at most once), 
and using a gadget shown in Figure~\ref{fig:hangleaves}, adapted from Fig. 4 of~\cite{ISS2010b} with arcs 
of length 1, and the operations described in the proof of Theorem 3 of the same article.}

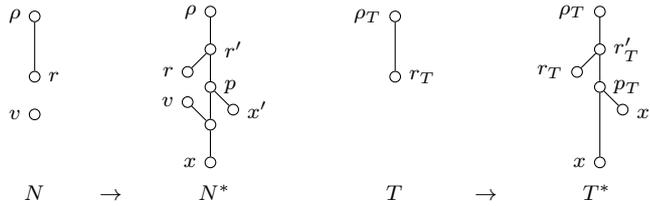
\begin{figure}[h]
 \begin{center}

\begin{tabular}{ccccccc}
\begin{tikzpicture}
\node[vertex,draw] (root) at (0,2) [label=left:$\rho$] {};
\node[vertex,draw] (rootc) at (0,1.2) [label=right:$r$] {};
\node[vertex,draw] (v) at (0,0.7) [label=left:$v$] {};
\node (void) at (0,0) {};

\draw (root) -- (rootc);
\end{tikzpicture}
& ~ &
\begin{tikzpicture}
\node[vertex,draw] (rootp) at (0,2) [label=left:$\rho$] {};
\node[vertex,draw] (root) at (0,1.5) [label=right:$r'$] {};
\node[vertex,draw] (rootc) at (-0.3,1.2) [label=left:$r$] {};
\node[vertex,draw] (p) at (0,1) [label=right:$p$] {};
\node[vertex,draw] (xp) at (0.3,0.7) [label=right:$x'$] {};
\node[vertex,draw] (v) at (-0.3,0.8) [label=left:$v$] {};
\node[vertex,draw] (px) at (0,0.5) {};
\node[vertex,draw] (x) at (0,0) [label=left:$x$] {};

\draw (rootp) -- (root);
\draw (root) -- (rootc);
\draw (root) -- (p);
\draw (p) -- (px);
\draw (p) -- (xp);
\draw (px) -- (x);
\draw (v) -- (px);
\end{tikzpicture}
& ~
&
\begin{tikzpicture}
\node[vertex,draw] (root) at (0,2) [label=left:$\rho_T$] {};
\node[vertex,draw] (rootc) at (0,1.2) [label=right:$r_T$] {};
\node (void) at (0,0) {};

\draw (root) -- (rootc);
\end{tikzpicture}
& ~ &
\begin{tikzpicture}
\node[vertex,draw] (rootp) at (0,2) [label=left:$\rho_T$] {};
\node[vertex,draw] (root) at (0,1.5) [label=right:$r'_T$] {};
\node[vertex,draw] (rootc) at (-0.3,1.2) [label=left:$r_T$] {};
\node[vertex,draw] (p) at (0,1) [label=right:$p_T$] {};
\node[vertex,draw] (xp) at (0.3,0.7) [label=right:$x'$] {};
\node[vertex,draw] (x) at (0,0) [label=left:$x$] {};

\draw (rootp) -- (root);
\draw (root) -- (rootc);
\draw (root) -- (p);
\draw (p) -- (x);
\draw (p) -- (xp);
\end{tikzpicture}
\\
$N$ & $\rightarrow$ & $N^*$ & ~ & $T$ & $\rightarrow$ & $T^*$
\end{tabular}

 \end{center}
\caption{\change{\textbf{How our slightly modified \textsc{HangLeaves}$(v)$
modifies $N$ and $T$.} Vertices $\rho$ and $\rho_T$ are the roots
of $N$ and $T$ respectively. All arcs have length 1, except
$(r',r)$ of $N^*$ which has the same length as $(\rho,r)$
of $N$,
$(r'_T,r_T)$ of $T^*$ which has the same length as $(\rho_T,r_T)$
of $T$ and $(p_T,x)$ which has length 2.}}
\label{fig:hangleaves}
\end{figure}

\begin{corollary}
\textsc{relaxed-TCBL} is strongly NP-complete, and \textsc{closest-TCBL} is strongly NP-hard.
\end{corollary}
\begin{proof}
TCBL can be easily reduced to both problems. Indeed, any  instance of TCBL corresponds to an instance of \textsc{relaxed-TCBL} with 
$m_T(a) = M_T(a) := \lambda_T(a)$ for each arc of $T$.
Additionally, TCBL can be reduced to \textsc{closest-TCBL} by checking whether there exists 
a solution $\widetilde{T}$ with $\max |\lambda_T(a) - \lambda_{\widetilde{T}}(\tilde{a})| = 0$.
\qed
\end{proof}

\subsection{Weak NP-completeness for level-2 networks\label{sec:level2Blob}}

The strong NP-completeness result above does not imply anything about the hardness of TCBL on networks of bounded level.
Unfortunately, TCBL is hard even for low-level networks, as we now show.

\begin{theorem}\label{TCBLlevel2}
TCBL is weakly NP-complete for level-2 binary networks.
\end{theorem}

\begin{proof} First, recall that \textsc{TCBL} is in NP (Theorem \ref{TCBLstrong}).
To prove the theorem, we will reduce from the \textsc{subset sum} problem: 
given a multiset of \change{positive} integers $I=\{n_1, \dots, n_k\}$  and a \change{positive} integer $s$,  
is there a non-empty subset of $I$ whose sum is $s$? 
The {\sc subset sum} problem is known to be weakly NP-complete.

Now, we show how to construct an instance of the TCBL problem with the required characteristics, for each instance of the {\sc subset sum} problem.
This can be done by defining the tree $T$ and the network $N$ as follows. The tree $T$ is defined as the rooted tree on two leaves labeled $a$ and $b$, parent $\rho'$ and root $\rho$, and arcs $(\rho,\rho')$, $(\rho', a)$ and  $(\rho',b)$, 
respectively of length 1, 1 and $s+3k+1$. The network $N$ is the network on the two leaves labeled $a$ and $b$ 
shown in Fig.\ \ref{fig:level2Blob}, where 
$L> s+3k+1$.
Then, it is easy to see that a positive instance of the TCBL problem gives a positive instance of the {\sc subset sum} problem through the previous transformation, and vice versa.
This is true because no switching $S$ of $N$ giving rise to $T$ will ever contain the arcs with length $L$.
Thus, the paths in $S$ going through the blob containing the arc with length $n_i$ can have either length 2 or $2+n_i$. 
Now, any path from $\rho'_N$ to the leaf labeled $b$ has to go through all blobs, and through all arcs connecting these blobs. 
The sum of the lengths of the arcs on this path but outside the blobs 
is $k+1$. 
Thus, there exists a path from 
$\rho'_N$ to $b$ with length $s+3k+1$ if and only if there is a non-empty subset of $I=\{n_1, \dots, n_k\}$ whose sum is $s$.

As to the weakness of this NP-completeness result, we refer to Section \ref{pseudoPoly}, where
we give a pseudo-polynomial algorithm for TCBL on any binary network of bounded level.
\qed 
\end{proof}

\pgfdeclarelayer{background}
\pgfsetlayers{background,main}

\tikzstyle{myline}=[line width=.6pt]

\def\myellipse#1#2{
  \begin{scope}
    \clip (#2) ellipse (1.3 and 1);
    \draw[line width=1.2pt, fill=white] (#2) ellipse (1.3 and 1);

    \coordinate (bottom_left) at ($(#2) - (1.3,1)$);
    \coordinate (top_right) at ($(#2) + (1.3,1)$);
    \draw[myline] (bottom_left) -- (top_right);

    \node at ($(#2)+(-.3,.2)$) {$L$};
    \node at ($(#2)+(-.45,-.7)$) {$2$}; 
  \end{scope}
  \node at ($(#2) + (60:1.3)$) {$1$};
  \node at ($(#2) + (-20:1.5)$) {$1$};
  \node at ($(#2) + (160:1.5)$) {$n_{#1}$};
}

\begin{figure}[htbp]
\begin{center}
  \begin{tikzpicture}
    \node[vertex,draw,label=below:$a$] (start) at (-3,-2) {};
    \node[vertex,draw,label=below:$b$] (end) at (0,-10.5) {};

    \foreach \i/\pos in {1/-2, 2/-5, k/-9} {
      \myellipse{\i}{0,\pos}
      \node[vertex,draw] at (0,\pos-1) {};
      \node[vertex,draw] at (0,\pos+1) {};
      \node[vertex,draw] at (-0.9,\pos-0.7) {};
      \node[vertex,draw] at (0.9,\pos+0.7) {};
    }
    \foreach \pos in {-3.4, -6.2, -7.8}{
      \node at (-.2,\pos) {$1$};
    }
    \foreach \pos in {-.6, -10.3}{
      \node at (.3,\pos) {$1$};
    }
    \draw[myline] (0,0) -- (0,0.5);
    \node[vertex,draw,label=right:$\rho'_N$] at (0,0) {};
    \node[vertex,draw,label=right:$\rho_N$] at (0,0.5) {};

    \node[vertex,draw] at (-3,-2) {};    

    \begin{pgfonlayer}{background}
      \draw[myline] (start) -- node[midway, yshift=6pt] {$1$} (0,0) -- (0,-6.4);
      \draw[myline, dotted] (0,-6.5) -- (0,-7.5);
      \draw[myline] (0,-7.6) -- (end);
    \end{pgfonlayer}
    
    \node at (-0.2,0.25) {1}; 

  \end{tikzpicture}

  \caption{\textbf{The network used in the proof of Theorem \ref{TCBLlevel2}.}
  }
\label{fig:level2Blob}
\end{center}
\end{figure}

\newpage

\section{Positive results\label{sec:pos}}

\subsection{TCBL is FPT in the level of the network when no blob is redundant} \label{sec:2outarcs}



Note that in the weak NP-completeness result from Section \ref{sec:level2Blob} 
the blobs have only one outgoing arc each -- that is, they are redundant.
If we require that every blob has at least two outgoing arcs,
then dynamic programming becomes possible, and the problem becomes much easier. The
high-level reason for this as follows. 
Because blobs in the network $N$ have at least two outgoing arcs,
the image of any tree displayed by $N$ will branch at least once inside each blob.
This means that for each arc $(u', v')$ of $T$, if the image of $u'$ lies inside a blob $B$,
then the image of $v'$ either  lies (i) also inside $B$ or (ii) in one of the biconnected components $C_i$
immediately underneath $B$. This last observation holds with or without arc lengths, but
when taking lengths into account it has an extra significance. Indeed, suppose $N$ displays $T$ 
taking lengths into account, and $S$ is a switching of $N$ that induces the image of $T$ inside $N$. 
Let $(u',v')$ be an arc of $T$. If, within $S$, the image of $u'$
lies inside a blob $B$ and the image of $v'$ lies inside a biconnected component $C_i$ immediately
underneath $B$, then the image of the arc $(u',v')$ -- a path -- is naturally partitioned into 3 parts. 
Namely, a subpath inside $B$ (starting at the image of $u'$), followed by an outgoing arc of $B$,
followed by a subpath inside $C_i$ (terminating at the image of $v'$). See Fig.\ \ref{fig:steven} for an illustration. 
Within $S$, the lengths of these 3 parts must sum to $\lambda_{T}(u',v')$. 
The dynamic programming algorithm described below, in which we process the blobs in a bottom-up fashion, makes heavy use of this insight.

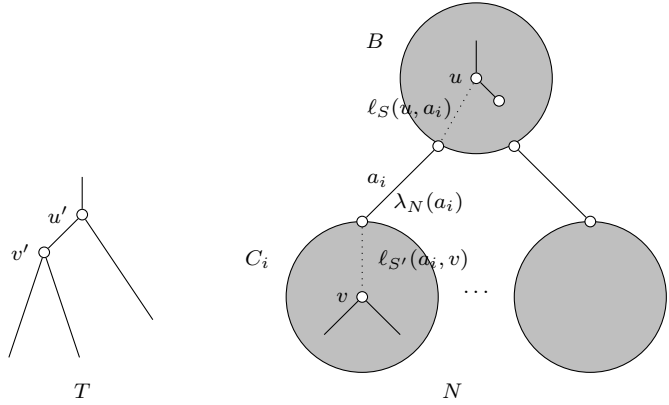
\begin{figure}[h]
 \begin{center}

\begin{tabular}{ccc}
\begin{tikzpicture}
\node[vertex,draw] (root) at (3,5) [label=left:$u'$] {};
\node[vertex,draw] (v0) at (2.5,4.5) [label=left:$v'$] {};

\node (v1) at (4,3.5) {};
\node (v2) at (2,3) {};
\node (v3) at (3,3) {};

\draw (3,5.5) -- (root);
\draw (root) -- (v0);
\draw (root) -- (v1);
\draw (v0) -- (v2);
\draw (v0) -- (v3);
\end{tikzpicture}
& ~ &
\begin{tikzpicture}

\draw[fill=lightgray,radius=1] (-2,1) circle ;
\node (ldots) at (-0.5,1.3) [label=below:$\ldots$] {};
\draw[fill=lightgray,radius=1] (1,1) circle ;
\draw[fill=lightgray,radius=1] (-0.5,3.9) circle ;
\node at (-1.5,4.4) [label=left:$B$] {};
\node[vertex,draw] (u) at (-0.5,3.9) [label=left:$u$] {};
\node[vertex,draw] (us) at (-0.2,3.6) {};
\node at (-0.6,3.5) [label=left:${\ell_{S}(u,a_{i})}$] {}; 
\node[vertex,draw] (vpp) at (-1,3) {};
\node at (-1.8,2.25) [label=right:$\lambda_{N}(a_i)$] {}; 
\node at (-1.8,2.25) [label=above:$a_i$] {}; 
\node[vertex,draw] (vpp2) at (0,3) {};
\node[vertex,draw] (vp2) at (1,2) {};
\node[vertex,draw] (vp) at (-2,2) {};
\node at (-3,1.5) [label=left:$C_i$] {};
\node at (-2,1.5) [label=right:${\ell_{S'}(a_{i},v)}$] {}; 
\node[vertex,draw] (v) at (-2,1) [label=left:$v$] {};

\draw (-0.5,4.4) -- (u);
\draw [dotted] (u) -- (vpp);
\draw (u) -- (us);
\draw (vpp2) -- (vp2);
\draw (vpp) -- (vp);
\draw [dotted] (vp) -- (v);
\draw (v) -- (-2.5,0.5);
\draw (v) -- (-1.5,0.5);

\end{tikzpicture}
\\
$T$ & ~ & $N$
\end{tabular}

 \end{center}
\caption{\textbf{Illustration of the idea at the basis of Algorithm \ref{alg:steven}.} If a network $N$ displays a tree $T$
and the image $u$ of $u'$ (for an arc $(u', v')$ of $T$) lies inside a blob $B$ of $N$, then -- assuming every blob of the network has at least two outgoing arcs -- the image $v$ of $v'$ will either lie inside $B$, or inside a blob $C_i$ that is immediately beneath $B$. In the latter case the image of $(u', v')$ can be naturally partitioned into three parts, as shown. This is the foundation
for the dynamic programming approach used in Theorem \ref{thm:levk} and later in Theorems \ref{thm:relaxedTCBL} and \ref{thm:closestTCBL}.}
\label{fig:steven}
\end{figure}

\begin{theorem}
\label{thm:levk}
Let $N$ be a level-$k$ binary network and $T$ be a rooted binary tree, both on $X$. If
no blob of $N$ is redundant,
then TCBL can be solved in time $O( k \cdot 2^k \cdot n )$ using $O( k \cdot 2^k \cdot n )$ space, where $n=|X|$.
\end{theorem}
\begin{proof}
Firstly, note that networks can have nodes that are not inside blobs (i.e. tree-like regions).
To unify the analysis, it is helpful to also regard such a node $u$ (including when $u$ is a taxon) as a blob with 0 reticulations: 
the definition of incoming and outgoing arcs extends without difficulty. Specifically, in this
case they will simply be the arcs incoming to and outgoing from $u$. We regard such blobs as having exactly one switching. 

Next, it is easy to see 
that the blobs of $N$ can themselves be organized as a rooted
tree, known as the \emph{blobbed-tree} \cite{GBP2009,Gusfield2014}.
In other words, the parent-child relation between blobs is well-defined, and unique. 
The idea is to process the blobs in bottom-up, post-order fashion. Hence, if a blob $B$ has blob
children $C_1, C_2, \ldots$ (underneath outgoing arcs $a_1, a_2 \ldots$) we first process $C_1, C_2, \ldots$ and then $B$. Our goal is to identify some switching of $B$ which can legitimately be merged with one switching each from $C_1, C_2, \ldots$. We initialize the dynamic programming
by, for each blob $B$ that is a taxon, recording that it has exactly one switching whose root-path
has length 0. (The definition and meaning of root-path will be given in due course).

For each blob $B$ that is not a taxon, we will loop through the (at most) $2^k$ ways
to switch the reticulations within $B$. Some of these candidate switchings can be
immediately discarded on topological grounds, 
i.e., such a switching of $B$ induces bifurcations that are not present in $T$.
Some other candidate switchings $S$ can be discarded on the basis of the lengths of their internal paths,
that is, the paths $u\rightarrow v$ entirely contained within $S$ coinciding with the image of some arc $(u',v')$ in $T$.
Clearly the path $u\rightarrow v$ must have the same length as $(u',v')$.

Finally, we need to check whether the candidate switching $S$ can be combined with switchings from
$C_1, C_2, \ldots$ such that arc lengths are correctly taken into account. This proceeds
as follows. Observe that, for each outgoing arc $a_i$ of $B$, $a_i$ lies 
on the image of an arc $(u', v')$ of $T$. This arc of $T$ is uniquely defined.
Let $u$ be the image of $u'$ in $B$, and let $\ell_{S}(u, a_i)$ be the total length of the path (in $S$)
from $u$ to the tail of $a_i$. The image of $v'$ will lie somewhere inside $C_i$. For
a switching $S'$ of $C_i$, let $v$ be the image of $v'$ within $S'$, and let $\ell_{S'}(a_i,v)$
be the total length of the path (in $S'$) from the head of $a_i$ to $v$. (See Fig.\ \ref{fig:steven}).

If we wish to
combine $S'$ with $S$, then we have to require $\lambda_T(u',v') = \ell_S(u,a_i) + \lambda_N(a_i) + \ell_{S'}(a_i,v)$.
To know whether such an $S'$ exists, $B$ can ask $C_i$ the question: ``do you have
a candidate switching $S'$ such that  $\ell_{S'}(a_i,v) = \lambda_T(u',v') - \ell_S(u,a_i) - \lambda_N(a_i)$ ?''
This will be true if and only if $C_i$ has a candidate switching $S'$ such that the
\emph{root-path} in $S'$ -- defined as the path from the root of $C_i$ to the first branching node 
of $S'$ -- has length exactly $\lambda_T(u',v') - \ell_S(u,a_i) - \lambda_N(a_i)$.
(We consider a node of $S'$ to be a branching node if it is the image of some node of $T$.)
$B$ queries all its children $C_1, C_2, \ldots$ in this way. If all the $C_i$ answer affirmatively, then
we store $S$, together with the length of its root-path, as a candidate switching of $B$,
otherwise we discard $S$.

This process is repeated until we have finished processing the highest blob $B$ of $N$. 
The answer to TCBL is YES, if and only if this highest blob $B$ has stored at least one candidate switching.
Pseudocode formalizing the description above is provided in Algorithm \ref{alg:steven}.


We now analyse the running time and storage requirements.  For step 1, observe that the blobbed-tree can easily be constructed once all the biconnected components of the undirected, underlying graph of $N$ have been identified. The biconnected components can be found in linear time (in the size of the graph) using the well-known algorithm of Hopcroft and Tarjan (see, e.g., \cite{Intro_to_Algorithms}).
Because every blob has at least two outgoing arcs,
$N$ will have $O(kn)$ vertices and arcs, 
(see, e.g., Lemma 4.5 in \cite{LeosThesis} and discussion thereafter)
so the time to construct the blobbed-tree is at most $O(kn)$.
Moreover, $N$ has $O(n)$ blobs, meaning that the blobbed-tree has $O(n)$ nodes and that
step 2 can be completed in $O(n)$ time by checking whether 
$T$ and the blobbed-tree are compatible. (The compatibility of two trees can be tested in linear
time \cite{warnow1994tree}.) Each blob has at most $2^k$ switchings, and each switch can be encoded
in $k$ bits. If we simply keep all the
switchings in memory (which can be useful for constructing an actual switching of $N$,
whenever the answer to TCBL is YES) then at most $O(k \cdot 2^k \cdot n)$ space is required. 

For time complexity, each blob $B$ loops through at most $2^k$ switchings, and for each switching $S$ it is necessary to check the topological legitimacy of $S$ (step 3(a)), that internal paths of the switching have the correct lengths (step 3(b)), and subsequently to make exactly one query to each of its child blobs $C_i$ (step 3(c)).  We shall return
to steps 3(a) and 3(b) in due course. It is
helpful to count queries from the perspective of the blob that is queried. In the entire
course of the algorithm, a blob will be queried at most $2^k$ times. 
Recalling that the number of blobs is $O(n)$, in total at most $O( 2^k \cdot n )$ queries will be made, 
so the total time devoted to queries is $O( q \cdot 2^k \cdot n )$, where $q$ is the time to answer each query. 
\change{Recall that a query consists of checking whether a blob has a switching whose root-path has a given length.}
Each blob needs to store at most $2^k$ switchings. By storing these switchings (ranked by the lengths of their
root-paths) in a balanced look-up structure (e.g. red-black trees) an incoming query can
be answered in \change{logarithmic time in the number of stored switchings, that is, in}
time $\log 2^k = k$. Hence, the total time spent on queries is $O( k \cdot 2^k \cdot n )$.


For steps 3(a) and 3(b) we require amortized analysis. Let $d^+(B)$ denote the number of outgoing arcs from blob $B$.
The blob $B$ can be viewed in isolation as a rooted phylogenetic network with $d^+(B)$ ``taxa'', so inside $B$ there are
$O(k \cdot d^+(B) )$ vertices and arcs \cite{LeosThesis}.
Therefore, the time to convert a switching $S$ from $B$ into a tree $T'$ on $d^+(B)$ ``taxa'' 
(via dummy leaf deletions and vertex smoothings)
is at most $O(k \cdot d^+(B) )$. 
The topology 
and internal arc lengths of $T'$ can be checked against those of the corresponding part of $T$ 
in $O(  d^+(B)  )$ time \cite{warnow1994tree}.
Hence, the total time spent on steps 3(a) and 3(b) is
\begin{equation}
\label{eq:amort}
 \sum_B O( 2^k \cdot  k \cdot d^+(B) ),
\end{equation}
where the sum ranges over all blobs. Note that $\sum_B  d^+(B)$ is $O(n)$ because there are $O(n)$ blobs and
each outgoing arc enters exactly one blob. Hence, the expression (\ref{eq:amort}) is $O( k \cdot 2^k  \cdot n )$,
matching the time bound for the queries. 
Hence, the overall running time of the algorithm is $ O( k \cdot 2^{k} \cdot n )$. \qed
\end{proof}

\begin{framed} 
\begin{algorithm} \label{alg:steven}
\normalfont
FPT algorithm for \textsc{TCBL} on networks with no redundant blobs
\begin{enumerate}
\item
Decompose $N$ into blobs and construct the blobbed-tree $T_N$, whose nodes are the blobs in $N$ and whose arcs are the arcs external to the blobs.

\item
Check that $T_N$ is compatible with the input tree $T$ 
(in fact check that $T_N$ can be obtained from $T$ via arc contractions). 
If this is not the case, then terminate with a NO. 
Otherwise 
each vertex $B$ in $T_N$ can be obtained as the contraction of a subtree $T(B)$ of $T$,
and each arc $a$ in the blobbed-tree $T_N$ originates from an arc $a'$ in $T$.
Store references to the $a'$ and the $T(B)$.

\item
\begin{description}
\item \textbf{for each} blob $B$, in bottom-up order:
  \begin{description}
    \item \textbf{for each} switching $S$ of $B$:

    \begin{enumerate}
      \item check that $S$ is topologically compatible with $T(B)$. \label{step:topocheck}

      \item \label{step:internal_checks}
      check that each arc of $T(B)$ is as long as its image in $S$; 

      \item \textbf{for each} blob $C_i$ that is a child of $B$, via the arc $a_i$: \label{step:foreach_blob}
      \begin{itemize}
        \item check that $C_i$ has stored a switching $S'$ 
        whose root-path has the appropriate length.
        Specifically, we require $\ell_{S'}(a_i,v) = \lambda_T(u',v') - \ell_S(u,a_i) - \lambda_N(a_i)$,
        where $(u', v')$ is the arc of $T$ on whose image $a_i$ lies (i.e. $a'_i$),
        and $u$ and $v$ are the uniquely defined images of $u'$ and $v'$ in $S$ and $S'$, respectively.
      \end{itemize}


      \item if none of the checks above failed, store $S$ along with 
      the length of its root-path; \label{step:what_to_store}

      \item if no switching is stored for $B$, then terminate with NO, as no tree displayed by $N$ satisfies the requirements. 
      \end{enumerate}

  \end{description}
\end{description}
\item If the algorithm gets this far, then it returns YES and the 
image in $N$ of $T$ can be obtained by combining a switching $S$ 
stored for the root blob, to the switchings $S'$ found for its child blobs, recursively. \label{step:conclusion}
\end{enumerate}
\end{algorithm}
\end{framed} 

\subsection{Pseudo-polynomial solution of TCBL on any network of bounded level\label{pseudoPoly}}

Redundant blobs are problematic for TCBL 
because when they appear ``in series'' (as in Fig.\ \ref{fig:level2Blob}) they give rise to an exponential explosion of paths that can be the images of an arc $a$ in $T$, and, as we saw, checking the existence of a path of the appropriate length $\lambda_T(a)$ is at least as hard as \textsc{subset sum}. Just like for \textsc{subset sum}, however, a pseudo-polynomial time solution is possible, as we now show.


\begin{theorem}
\label{thm:pseudopoly_TCBL}
Let $N$ be a level-$k$ binary network with $b$ blobs, and let $T$ be a rooted binary tree on the same set of $n$ taxa.
TCBL can be solved in time $O( k \cdot b \cdot L + 2^k \cdot n \cdot L )$ using $O( k \cdot n \cdot L )$ space, 
where $L$ is an upper bound on arc lengths in $T$. 
\end{theorem}
\begin{proof}
The algorithm we now describe is based on the following two observations 
(we use here the same notational conventions as in Algorithm \ref{alg:steven}).
First, if $T$ is indeed displayed by $N$, the image $u\rightarrow v$ of any of
its arcs $(u',v')$ will either be entirely contained in one blob, or $u$ and 
$v$ will be in two different blobs, which can only be separated by redundant blobs.
Second, it only makes sense to store a switching $S'$ of a blob $C_i$, if 
$\ell_{S'}(a_i,v) < \lambda_T(u',v')$, i.e., if its root-leaf path is shorter
than the corresponding arc in $T$, meaning that we only need to store $O(L)$ switchings per blob.

Accordingly, we modify Algorithm \ref{alg:steven} as follows:
\begin{description}
\item[Step 2a:] Check that $T_N$ is compatible with the input tree $T$ in the following way:
replace any chain of redundant blobs $M_1,M_2,\ldots,M_h$ in $T_N$ with a single arc 
from the parent of $M_1$ to the child of $M_h$,
and then check that the resulting blobbed-tree $T'_N$ can be obtained from $T$ via arc contractions.
If this is not the case then terminate with a NO. 
Otherwise for each arc $a$ and vertex $B$ in $T'_N$, define and store $a'$ and $T(B)$ as before.
\item[Step 2b:] 
For each arc $a$ in $T'_N$, calculate a set of lengths $L(a)$ as follows.
If $a$ originates from a chain of redundant blobs $M_1,M_2,\ldots,M_h$, then $L(a)$ is obtained 
by calculating the lengths of all paths in $N$ starting with the incoming arc of $M_1$ 
and ending with the (unique) outgoing arc of $M_h$. Only keep the lengths that are smaller than $\lambda_T(a')$.
For the remaining arcs in $T'_N$, simply set $L(a):=\{\lambda_N(a)\}$.
\end{description}
The algorithm only visits non-redundant blobs, performing a bottom-up traversal of $T'_N$, 
and doing the same as Algorithm \ref{alg:steven} except for:
\begin{description}
\item[Step \ref{step:foreach_blob}:] \textbf{for each} blob $C_i$ that is a child of $B$ in $T'_N$ :
\begin{itemize}
\item check the existence of an $\ell\in L(a_i)$ and a switching $S'$ stored for $C_i$ that satisfy:
\begin{equation}
  \ell_{S'}(a_i,v) = \lambda_T(u',v') - \ell_S(u,a_i) - \ell. \label{eqn:pseudopoly_TCBL}
\end{equation}
\end{itemize}
\end{description}

To complete the description of the algorithm resulting from these changes, we assume that the switchings for a (non-redundant) blob $B$ are stored in an array $S_B$ indexed by the root-path length of the switching. 
If two or more switchings of a blob have the same root-path length $\ell$, we only keep one of them in $S_B[\ell]$.
Because for $C_i$ we only store the switchings whose root-path length is less than $\lambda_T(a'_i)$, the $S_B$ arrays have
size $O(L)$.

As for step 2b above, the computation of $L(a)$ for an arc in $T'_N$ corresponding to a chain of redundant blobs can be implemented in a number of ways. Here we assume that the vertices in $M_1,M_2,\ldots,M_h$ are visited following a topological ordering, and that, for each visited vertex $v$, we fill a boolean array $P_v$ of length $\lambda_T(a')$, where $P_v[\ell]$ is true if and only if there exists a path of length $\ell$ from the tail of the arc incoming $M_1$ to $v$. Once $P_{v_h}$ for the head $v_h$ of the arc outgoing $M_h$ has been filled, $L(a)$ will then be equal to the set of indices $\ell$ for which $P_{v_h}[\ell]$ is true.

We are now ready to analyse the complexity of this algorithm. We start with the space complexity. First note that every redundant blob of level $k$ in a binary network must have exactly $2k$ vertices (as the number of bifurcations must equal the number of reticulations). Because each redundant blob has $O(k)$ vertices, and each $P_v$ array is stored in $O(L)$ space, step 2b requires $O(kL)$ space to process each redundant blob. Because every time a new redundant blob $M_{i+1}$ is processed, the $P_v$ arrays for the vertices in $M_{i}$ can be deleted, step 2b only requires $O(kL)$ space in total. This however is dominated by the space required to store $O(L)$ switchings for each non-redundant blob. Since there are $O(n)$ non-redundant blobs in $N$ and each switching requires $O(k)$ bits to be represented, the space complexity of the algorithm is $O( k \cdot n \cdot L )$. 

We now analyse the time complexity. Checking the compatibility of the blobbed tree and $T$ (step 2a) can be done in time $O(n+b)$, as this is the size of $T'_N$ before replacing the chains of redundant blobs. The computation of the arrays $P_v$ (step 2b) involves $O(L)$ operations per arc in $M_1,M_2,\ldots,M_h$. Because there are $O(b)$ redundant blobs, and because each of them contains $O(k)$ arcs, calculating all the $P_v$ arrays requires time $O( k \cdot b \cdot L)$.

The other runtime-demanding operations are the queries in step \ref{step:foreach_blob}. 
These involve asking, for each $\ell\in L(a_i)$, whether $C_i$ has a switching whose 
root-path has the length in Eqn.\ (\ref{eqn:pseudopoly_TCBL}). 
Each of these queries can be answered 
in constant time by checking whether $S_{C_i}[\lambda_T(u',v') - \ell_S(u,a_i) - \ell\,]$ is filled or not. 
Because every non-redundant blob $C_i$ will be queried at most $2^k\cdot L(a_i)$ times, 
and because there are $O(n)$ non-redundant blobs, the total time devoted to these queries 
is $O(2^k \cdot n \cdot L)$.
The remaining steps require the same time complexities as in Theorem \ref{thm:levk}. By adding up all these runtimes we obtain a total time complexity of $O( k \cdot b \cdot L + 2^k \cdot n \cdot L )$.
\qed
\end{proof}

\subsection{\textsc{closest-TCBL} and \textsc{relaxed-TCBL} are FPT in the level of the network when no blob is redundant} \label{sec:FPT2relProblems}

We now show that Algorithm \ref{alg:steven} can be adapted to solve the ``noisy'' variations of \textsc{TCBL} that we have introduced in the Preliminaries section.

\begin{theorem}
\label{thm:relaxedTCBL}
Let $N$ be a level-$k$ binary network and $T$ be a rooted binary tree, both on the same set of $n$ taxa. 
The arcs of $N$ are labelled by positive integer lengths, and the arcs of $T$ are labelled by a minimum and a maximum positive integer length.
If no blob of $N$ is redundant, then \textsc{relaxed-TCBL} can be solved in $O(k \cdot 2^k \cdot n )$ time and space.
\end{theorem}
\begin{proof}
We modify Algorithm \ref{alg:steven} to allow some flexibility whenever a check on lengths is made: 
instead of testing for equality between arc lengths in the tree and the
path lengths observed in the partial switching under consideration, 
we now check that the path length belongs to the input interval.
Specifically, we modify two steps in Algorithm \ref{alg:steven} as follows:
\begin{description}
\item[Step \ref{step:internal_checks}:] check that every arc $(u',v')$ of $T(B)$, 
whose image is an internal path $u\rightarrow v$ of $S$, is such that $\ell_S(u,v) \in [m_T(u',v'), M_T(u',v')]$.
\item[Step \ref{step:foreach_blob}:] check that, among the switchings stored for $C_i$, there exists at least one switching $S'$ whose root-path $\ell_{S'}(a_i,v)$ has a length in the appropriate interval. Specifically, using the same notation as in Algorithm \ref{alg:steven}, check that:
\[
  \ell_S(u,a_i) + \lambda_N(a_i) + \ell_{S'}(a_i,v) \quad \in \quad [m_T(a_i'), M_T(a_i')],
\]
that is:
\begin{equation} \label{eq:interval}
  m_T(a_i') - \ell_S(u,a_i) - \lambda_N(a_i) \quad\le\quad \ell_{S'}(a_i,v)  \quad\le\quad M_T(a_i') - \ell_S(u,a_i) - \lambda_N(a_i).
\end{equation}
\end{description}
We can use the same data structures used by Algorithm \ref{alg:steven}, so the space complexity remains $O(k \cdot 2^k \cdot n )$. As for time complexity, the only relevant difference is in step \ref{step:foreach_blob}: instead of querying about the existence of a switching with a definite path-length, we now query about the existence of a switching whose path-length falls within an interval (see Eqn.\ (\ref{eq:interval})). In a balanced look-up structure, this query can be answered again in time $O(\log 2^k)=O(k)$.
In conclusion the time complexity remains the same as that in Theorem \ref{thm:levk}, that is $O(k \cdot 2^k \cdot n )$. 
\qed
\end{proof}

\begin{theorem}
\label{thm:closestTCBL}
Let $N$ be a level-$k$ binary network and $T$ be a rooted binary tree, both with  \change{positive}  integer arc lengths and on the same set of $n$ taxa. 
If no blob of $N$ is redundant, then \textsc{closest-TCBL} can be solved in
time $O( 2^{2k} \cdot n )$ using $O( k \cdot 2^k \cdot n )$ space.
\end{theorem}
\begin{proof}
As we shall see, we modify Algorithm \ref{alg:steven} by removing all checks on arc lengths, and by keeping references to those switchings that may become part of an optimal solution in the end: any topologically-viable switching $S$ of a blob $B$ is stored along with a reference, for each child blob $C_i$, to the switching $S'$ that must be combined with $S$. Moreover, we compute recursively $\mu_S$, which we define as follows: 
\[
  \mu_S = \max \left| \lambda_T(a) - \lambda_{\widetilde{T}}(\tilde{a}) \right|,
\]
where $\widetilde{T}$ is the subtree displayed by $N$ obtained by (recursively) combining $S$ to the switchings stored for its child blobs, and then applying dummy leaf deletions and vertex smoothings. The $\max$ is calculated over any arc $\tilde{a}$ in $\widetilde{T}$ and its corresponding arc $a$ in $T$, excluding the root arc $\tilde{a}_r$ of $\widetilde{T}$ from this computation. This is because the length of the path above $S$, which must be combined with $\tilde{a}_r$,
is unknown when $S$ is defined.

In more detail, we modify Algorithm \ref{alg:steven} as follows:

\begin{description}
\item[Step \ref{step:internal_checks}:] no check is made on the lengths of the internal paths of $S$; instead initialize $\mu_S$ as follows:
\[
\mu_S := \max \left| \lambda_T(u',v') - \ell_S(u,v) \right|,
\]
where the $\max$ is over all arcs $(u',v')$ in $T(B)$, and $u, v$ are the images of $u', v'$ in $S$, respectively.
Trivially, if $B$ is just a vertex in $N$, the $\max$ above is over an empty set, meaning that $\mu_S$ can be initialized 
to any sufficiently small value (e.g., 0).

\item[Step \ref{step:foreach_blob}:] \textbf{for each} blob $C_i$ that is a child of $B$:
\begin{itemize}
\item among the switchings stored for $C_i$, seek the switching $S'$ minimizing
\begin{equation} \label{eq:optimality_principle}
\max \left\{ \mu_{S'}, \left| \ell_S(u,a_i) + \lambda_N(a_i) + \ell_{S'}(a_i,v) - \lambda_T(u',v') \right| \right\}
\end{equation}
\item set $\mu_S$ to $\max \left\{ \mu_S, \mbox{value of (\ref{eq:optimality_principle}) for } S' \right\}$
\end{itemize}
\item[Step \ref{step:what_to_store}:] 
store $S$ along with $\mu_S$, with the length of its root-path, and with references to the child switchings $S'$ minimizing (\ref{eq:optimality_principle})
\item[Step \ref{step:conclusion}:] seek the switching $S$ stored for the root blob that has minimum $\mu_S$, and combine it recursively to the switchings $S'$ found for its child blobs. In the end,  a switching $\widetilde{S}$ for the entire network $N$ is obtained, which can be used to construct $\widetilde{T}$.
\end{description}

The correctness of the algorithm presented above is based on the following observation, allowing our dynamic programming solution of the problem: 

\textbf{Observation.}
Let $B$ be a blob of $N$, and $C_i$ be one of its child blobs. If a switching $S$ of $B$ is part of an optimal solution to \textsc{closest-TCBL}, then we can assume that $S$ must be combined with a switching $S'$ of $C_i$ that minimizes  Eqn.\ (\ref{eq:optimality_principle}). This means that even if there exists an optimal solution in which $S$ is combined with $S''$, a non-minimal switching of $C_i$ with respect to Eqn.\ (\ref{eq:optimality_principle}), then we can replace $S''$ with $S'$ and the solution we obtain will still be optimal. 

Once again, space complexity is $O(k \cdot 2^k \cdot n )$, 
as the only additional objects to store are the references to the child switchings of $S$, 
and the value of $\mu_{S}$ for each the $O(2^k \cdot n)$ stored switchings.
As for time complexity, each query within step \ref{step:foreach_blob} now involves scanning the entire set of $O(2^k)$ switchings stored for $C_i$, thus taking time $O(2^k)$ -- whereas the previous algorithms only required $O(k)$ time. Since there are again $O(2^k \cdot n)$ queries to make, the running time is now $O( 2^{2k} \cdot n )$.
\qed
\end{proof}

We conclude this section by noting that, if we reformulate \textsc{closest-TCBL} replacing the $\max$ with a 
sum, and taking any positive power of the absolute value $\left| \lambda_T(a) - \lambda_{\widetilde{T}}(\tilde{a}) \right|$ 
in the objective function, then the resulting problem can still be solved in a way analogous to that described above.
\begin{theorem}
\label{thm:reformulations}
Consider the class of minimization problems obtained from \textsc{closest-TCBL} by replacing its objective function with
\begin{equation}
\biguplus_a \left| \lambda_T(a) - \lambda_{\widetilde{T}}(\tilde{a}) \right|^d,
\end{equation}
with $\biguplus$ representing either $\max$ or $\sum$, 
and with $d>0$. 

If no blob of $N$ is redundant, then any of these problems can be solved 
in time $O( 2^{2k} \cdot n )$ using $O( k \cdot 2^k \cdot n )$ space, where $n$ is the number 
of taxa in $N$ and $T$, and $k$ is the level of $N$. 
\end{theorem}
\begin{proof}
In the proof of Theorem \ref{thm:closestTCBL}, replace every occurrence of $|\ldots|$ with $|\ldots|^d$, and -- if $\biguplus$ represents a sum -- replace every occurrence of $\max$ with $\sum$.
\qed
\end{proof}

It is worth pointing out that the algorithm described in the proofs above requires storing \emph{all} switchings of a blob that are topologically compatible with the input tree. This is unlike the algorithms shown before, where a number of checks on arc lengths (quite stringent ones in the case of the algorithm for TCBL) ensure that, on realistic instances, the number of switchings stored for a blob with $k$ reticulations will be much smaller than $2^k$.

Moreover, again unlike the previous algorithms, the queries at step 
\ref{step:foreach_blob} involve considering \emph{all} switchings stored 
for a child blob $C_i$, which is what causes the factor $2^{2k}$ in the
runtime complexity. 
We note that, for certain objective functions, it might be possible to make
this faster (with some algorithmic effort), but in order to achieve the 
generality necessary for Theorem \ref{thm:reformulations}, we have opted
for the simple algorithm described above.




\section{Discussion}

In this paper, we have considered the problem of determining whether a tree is
displayed by a phylogenetic network, when branch lengths are available.
We have shown that, if the network is permitted to have redundant blobs 
(i.e. nontrivial biconnected components with only one outgoing arc), 
then the problem becomes hard when at least one of the following two conditions hold: 
(1) the level of the network is unbounded (Theorem \ref{TCBLstrong}), 
(2) branch lengths are potentially long (Theorem \ref{TCBLlevel2}). 
If neither condition holds (i.e. branch lengths are
short and level is bounded) then -- even when redundancy is allowed -- 
the problem becomes tractable (Theorem~\ref{thm:pseudopoly_TCBL}).
We note that phylogenetic networks with redundant blobs are
unlikely to be encountered in practice, as their reconstructability
from real data is doubtful \cite{trinets1,trinets2,pardi2015reconstructible}. 
This is relevant because, if redundant blobs are not permitted, 
the problem becomes fixed-parameter tractable in the level of the
network (Theorem~\ref{thm:levk}) \emph{irrespective} of how long the branches are.


Building on our result on networks with no redundant blobs, 
we have then shown how the proposed strategy can be extended
to solve a number of variants of the problem
accounting for uncertainty in branch lengths.
This includes the case where an interval of possible lengths is
provided for each branch of the input tree (Theorem~\ref{thm:relaxedTCBL}),
and the case where we want to find -- among all trees displayed by the
network with the same topology as 
the input tree $T$ -- one that is 
closest to $T$, according to a number of measures of discrepancy between branch 
length assignments (Theorems \ref{thm:closestTCBL} and \ref{thm:reformulations}).

The fixed parameter algorithms we present here have runtimes and storage
requirements that grow exponentially in the level of the network.
However, in the case of storage, this is a worst-case scenario:
in practice, this will depend on the number of ``viable'' switchings stored 
for each blob, that is, the switchings that pass all checks on topology and 
branch lengths. 
In the case of the algorithm for \textsc{TCBL} (Theorem~\ref{thm:levk}),
where strict equalities between arc lengths in $T$ and path lengths in $N$ 
must be verified, we can expect it to be very rare that multiple switchings 
will be stored for one blob.
%
Similarly, in the case of the algorithm for \textsc{relaxed-TCBL} (Theorem~\ref{thm:relaxedTCBL}),
when the input intervals are sufficiently small, we can expect the number of 
stored switchings to be limited.
In some particular cases, 
it might even be possible to find the few viable switchings
for a blob, without having to consider all $O(2^k)$ switchings,
thus removing this factor from the runtime complexity as well.


The algorithm for \textsc{TCBL} (Algorithm \ref{alg:steven}) provides
a good example of the effect of taking into account branch lengths in
the tree containment problems:
if all checks on branch lengths are removed, 
what is left is an algorithm that solves the classic (topology-only)
tree containment problem, 
and also provides all ways to locate the input tree in the network
(for each blob, it can produce a list of possible images of the 
corresponding part of of the input tree).
This algorithm may run a little faster
than Algorithm \ref{alg:steven} (as no queries to child blobs are necessary).
However, for a small computational overhead, including branch lengths
allows to locate more precisely the displayed trees, and provides 
more strict answers to the tree containment problem.

\bibliographystyle{spmpsci}
\bibliography{references}


\end{document}